\documentclass[journal]{IEEEtran}

\usepackage{amsmath,amssymb,amsthm,amsfonts,bm}
\usepackage{graphicx}
\usepackage{caption}
\usepackage{subfig}
\usepackage{stfloats}
\usepackage{cite}
\usepackage{color,soul}
\usepackage{multirow}
\usepackage{rotating}
\usepackage{verbatim}
\usepackage{tabularx}
\usepackage{array}
\usepackage{url}
\usepackage{ragged2e}
\usepackage{xcolor,colortbl}
\usepackage{pifont}
\usepackage{balance}
\usepackage{algorithm}
\usepackage{algpseudocode}
\usepackage{orcidlink} 
\usepackage{cleveref}
\usepackage{enumitem}
\usepackage{booktabs}   
\usepackage{longtable}  

\definecolor{Gray}{gray}{0.85}

\newtheorem{theorem}{Theorem}
\newtheorem{lemma}[theorem]{Lemma}

\theoremstyle{definition}

\theoremstyle{remark}
\newtheorem{remark}[theorem]{Remark}

\newtheorem*{lemma*}{Lemma}

\begin{document}

\title{Square-Root Law for Covert Communication with Warden-Favorable Side Information}

\author{Hossein Ahmadi\,\orcidlink{0000-0002-3650-0280},
        Christian Deppe\,\orcidlink{0000-0002-2265-4887},
        Boulat A. Bash\,\orcidlink{0000-0002-1205-3906},
        and Eduard A. Jorswieck\,\orcidlink{0000-0001-7893-8435}%
\thanks{H. Ahmadi is with the Department of Electronic and Telecommunications, Politecnico di Torino, Italy (e-mail: hossein.ahmadi@polito.it).}%
\thanks{Christian Deppe and E. A. Jorswieck are with the Institute for Communications Technology, Technische Universität Braunschweig, Germany (e-mails: christian.deppe@tu-braunschweig.de, e.jorswieck@tu-braunschweig.de).}%
\thanks{Boulat A. Bash is with Dept. of Electrical and Computer Engineering, University of Arizona, Tucson, AZ, USA (e-mail: boulat@arizona.edu).}%
\thanks{Corresponding author: Hossein Ahmadi.}}

\maketitle
\begin{abstract}
Covert communication enables Alice to transmit to Bob while making the transmission difficult for Willie to detect. We study a scalar Gaussian covert-overlay model in which Alice's low-power covert signal is superimposed on an aggregate public component generated by Alice or other trackable sources. Willie is given all physically obtainable side information, including protocol details, timing, pilots, channel estimates, and calibration information, and subtracts his best estimate of the public component before testing. Covertness is imposed on the resulting residual through a relative-entropy constraint with budget $\delta$ conditioned on Willie's side information. In the stationary case, the residual under no covert transmission has variance $\sigma_0^2=\sigma_W^2+\sigma_e^2$, where $\sigma_W^2$ is Willie's receiver-noise variance and $\sigma_e^2$ is the irreducible cancellation error. Over $n$ channel uses, the maximal reliably transmissible covert payload is $R_C^\star\sqrt{n}(1+o(1))$ bits, where $R_C^\star=\frac{\sigma_0^2}{\sigma_B^2\ln 2}\sqrt{\delta}$, and $\sigma_B^2$ is Bob's receiver-noise variance. Thus, the square-root-law (SRL) constant is governed by the variance at Willie's actual detector input, not by receiver noise alone. Low-power Gaussian signaling achieves this constant, and a matching converse establishes first-order optimality within the conditioned additive Gaussian innovation model. For known time-varying conditioned residual variances, we also derive the first-order allocation, which assigns more covert power to larger residual variances. The results require a Gaussian post-cancellation null residual with known conditioned variance; non-Gaussian residuals and fixed non-vanishing variance uncertainty are outside the scope of this paper.
\end{abstract}

\begin{IEEEkeywords}
Covert communication, low-probability-of-detection communication, relative entropy, square-root law, residual floor.
\end{IEEEkeywords}
\IEEEpeerreviewmaketitle

\section{Introduction}\label{sec:intro}
\IEEEPARstart{C}{overt} communication, also called low-probability-of-detection (LPD) communication, aims to convey information while making the transmission statistically indistinguishable from noise to an adversary, traditionally called Warden Willie. In memoryless settings with a non-redundant no-input symbol, the classical square-root law (SRL) limits the covert payload to $\Theta(\sqrt{n})$ bits over $n$ channel uses \cite{Bash2012, Bash2013}. Accordingly, much work refines SRL constants and operating regimes: resolvability-based designs tighten message/key scaling \cite{Wang2016, Bloch2016}, finite-blocklength analyses characterize first- and second-order limits \cite{Tahmasbi2019}, and broadcast or multiple-access channel (MAC) extensions establish the optimality of time-division and simple budget splitting at SRL scales \cite{Tan2018, Arumugam2019}. Gaussian inputs and covariance shaping also characterize fundamental limits for multi-antenna additive white Gaussian noise (AWGN) links under LPD constraints \cite{Abdelaziz2017}.

Other work relaxes Willie’s knowledge or enriches the environment. Timing uncertainty enlarges the hypothesis space and yields $O\!\left(\min\{\sqrt{n\log T(n)},\,n\}\right)$ covert bits when the transmitter, Alice, selects a secret slot among $T(n)$ candidates
\cite{Bash2016}. Uninformed jammers and artificial noise (AN), whether deliberately injected or arising from channel uncertainty, increase Willie’s effective variance and can enable non-vanishing covert power or even positive rates under specific dynamics
\cite{Bash2017, Bash2018, Wang2019, He2017, Shahzad2017, Lee2014-1, Lee2014-2, Sobers2015}. Effective secrecy instead combines message secrecy and stealth through divergence from a default behavior; if that behavior can be chosen or shaped,
positive-rate stealth may be possible, unlike the classical SRL setting where Alice must hide in noise \cite{Hou2014,Hou2017}.

Multiuser structures offer further gains. Selecting one of multiple overt channels yields detection-error benefits \cite{Kim2022}, while non-orthogonal multiple access (NOMA) favors embedding with the hardest-to-decode overt user, improving spectral efficiency
under covertness constraints \cite{Ta2022}. Rate-splitting multiple access (RSMA) provides additional degrees of freedom to jointly serve overt and covert users under quality-of-service (QoS) and LPD constraints \cite{Zhang2024}. Recent mixed covert/non-covert
multiple-access results characterize simultaneous non-covert rates, covert SRL rates, and secret-key budgets, showing that non-covert users can enlarge the achievable covert-rate region under suitable coding and multiplexing \cite{Bounhar2024ICC,Bounhar2024arxiv}.

Networked and hierarchical settings combine covertness with secrecy and reliability. Untrusted relays and multi-warden environments yield nonconvex power-allocation problems and trade-offs between secrecy rate and detectability \cite{Forouzesh2020, Hu2018}. Distributed helper selection in slow fading improves covert throughput over single-jammer baselines \cite{Zheng2021}, while unmanned aerial vehicle (UAV)-aided uplinks use joint beamforming, trajectory, and multiuser design to characterize secret/covert-rate Pareto frontiers under outage and detection constraints \cite{Xu2025}. Together, these approaches provide a growing toolset based on timing, jamming/AN, fading uncertainty, and multiantenna/multiuser structure for approaching SRL limits \cite{Chen2023, Bash2015}.

Existing uncertainty- and side-information-based approaches, including timing \cite{Bash2016}, jammers/AN \cite{Bash2017, Bash2018, Wang2019, He2017, Lee2014-1,Lee2014-2, Sobers2015, Ahmadi2023}, fading, and channel state information (CSI) \cite{Shahzad2017, Lee2018}, generally treat overt traffic as incidental camouflage or part of a broader uncertainty environment. Effective-secrecy formulations compare communication with a default behavior and relate stealth to a divergence criterion \cite{Hou2014,Hou2017}, while mixed covert/non-covert multiple-access models show that non-covert users can enlarge the covert-rate region through joint coding and multiplexing \cite{Bounhar2024ICC,Bounhar2024arxiv}. By contrast, we study a warden-favorable innovation regime in which Willie uses all physically obtainable side information to perform the best block-level cancellation of the overt component before testing the resulting innovation residual. Covertness is imposed directly on the residual laws, conditioned on Willie’s side information, through a detector-agnostic relative-entropy constraint.

We focus on the SRL regime after Willie cancels the overt component as well as possible. The relevant noise level is therefore the variance of the residual actually tested under no covert transmission, not Willie’s raw receiver noise alone. In the stationary case, this variance is $\sigma_0^2=\sigma_W^2+\sigma_e^2$, where $\sigma_W^2$ is Willie’s receiver-noise variance and $\sigma_e^2$ is the irreducible cancellation-error variance. Our analysis requires this residual variance to be known after conditioning on Willie’s side information and excludes fixed non-vanishing uncertainty caused, for example, by untracked fading, interference, jamming, or power uncertainty. Under the residual-based relative-entropy constraint, the largest covert message size scales as $R_C^\star\sqrt{n}(1+o(1))$ bits, where the first-order constant $R_C^\star$, stated in Section~\ref{sec:main}, is governed by $\sigma_0^2$.

Our contributions are as follows:
\begin{enumerate}

\item \textbf{Warden-favorable innovation model and SRL regime boundary.} We formalize Willie’s side information, allow arbitrary block-level overt cancellation using all physically obtainable information, and impose covertness on the resulting residual. The regime requires a public/trackable overt component, no hidden-resource overt design, and a zero-mean Gaussian post-cancellation residual under no covert transmission with known conditioned variance. This separates the intended SRL regime from noise/channel-uncertainty \cite{He2017,Shahzad2017,Lee2014-1,Lee2014-2} and effective-secrecy regimes \cite{Hou2014,Hou2017} (Sections~\ref{sec:Model} and \ref{sec:regimes}).

\item \textbf{Relative-entropy budget and covert-power allocation.} We translate the conditioned relative-entropy budget into covert-power constraints. In the stationary case, an exact uniform allocation exhausts the budget, while a simpler closed-form approximation achieves the same first-order SRL constant. For time-varying residual variances, a small-signal expansion yields an ellipsoidal constraint and a first-order optimal scheduler that allocates more power to larger innovation-domain residual variances,
together with the corresponding heterogeneous payload scaling. We also provide a conservative pre-residualization design that imposes covertness on Willie’s full pre-test information (Section~\ref{sec:kl-equalizer}).

\item \textbf{Achievability and converse under the same conditioned relative-entropy constraint.} Using independent and identically distributed (i.i.d.) Gaussian covert symbols with uniform per-use variance, we prove achievability under the conditioned post-subtraction constraint. Under the same additive conditioned Gaussian innovation model, a matching first-order converse uses Gaussian extremality to show that Gaussian residuals minimize relative entropy for a fixed coordinate-variance profile; hence, sparse or on-off signaling cannot improve the first-order constant. This identifies the exact $\sqrt n$-scaling constant in the stationary regime (Sections~\ref{sec:Achievability} and \ref{sec:Converse}, and Theorem~\ref{thm:main}).

\end{enumerate}

The rest of the paper is organized as follows. Section~\ref{sec:problem} formalizes the model, side information, residualization, covertness constraint, and objective. Section~\ref{sec:kl-equalizer} derives the power equalizer and heterogeneous scheduler; Section~\ref{sec:regimes} interprets the innovation model and its boundary to linear-law/effective-secrecy regimes; Sections~\ref{sec:Achievability} and \ref{sec:Converse} establish achievability and the converse, assembled in the main theorem in Section~\ref{sec:main}; Section~\ref{sec:results} presents numerical results; and Sections~\ref{sec:future} and \ref{sec:conclusion} give extensions and conclusions.

\section{Problem Statement}\label{sec:problem}

This section provides the setting and notation for a scalar additive Gaussian observation model involving Bob's and Willie's received signals, an aggregate public component, and a low-power covert component generated by Alice. The main notation used in this section and throughout the paper is summarized in Table~\ref{tab:notations}. The key modeling step is that Willie first forms an innovation residual by performing block-level cancellation using his observation and all physically obtainable side information. Covertness is then enforced by constraining the relative entropy between the residual distributions under covert transmission and no covert transmission. Figure~\ref{fig:system} gives the operational picture, while this section sharpens the admissible SRL/innovation regime and delineates it from out-of-scope linear-law/effective-secrecy settings.

\begin{figure*}[t]
\centering
\includegraphics[width=\linewidth]{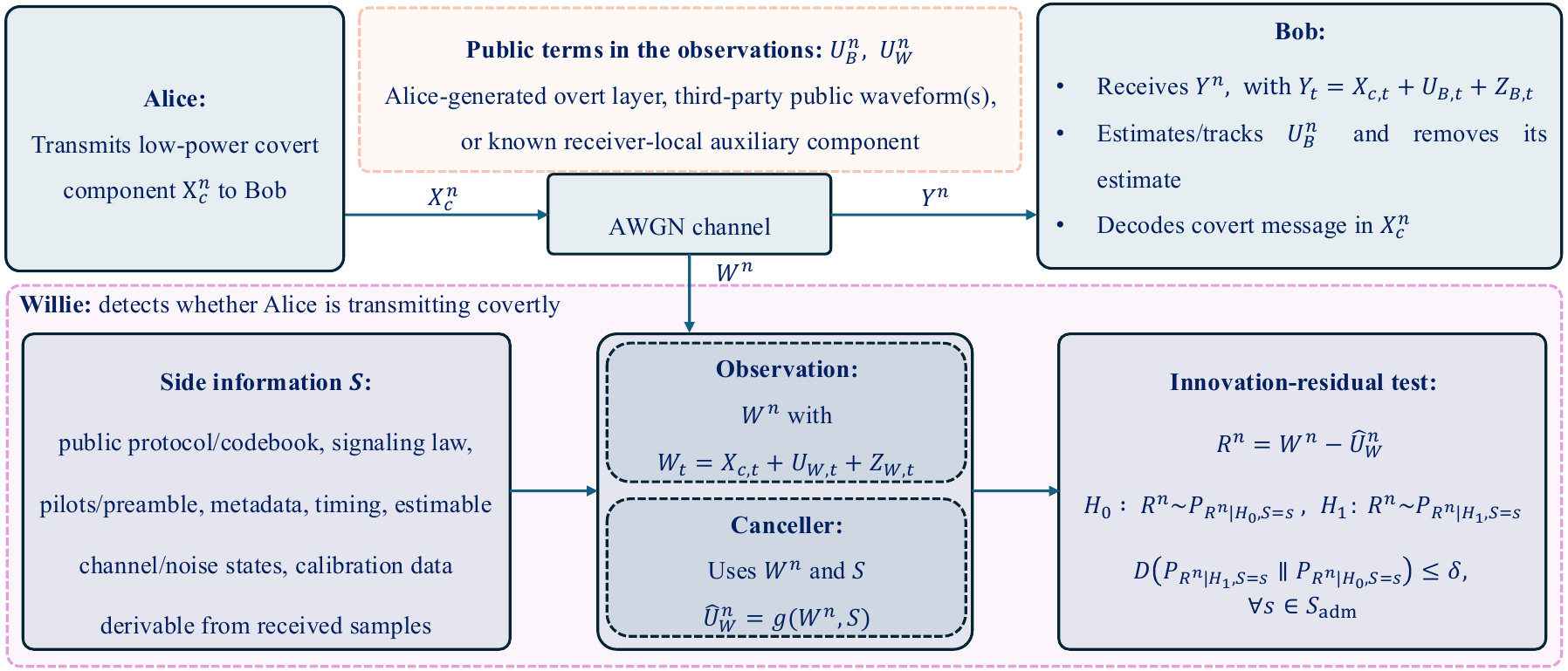}
\caption{Warden-favorable AWGN model under the conditioned post-subtraction relative-entropy constraint. Alice transmits the low-power covert component $X_c^n$ to Bob, while aggregate public/trackable components $U_B^n$ and $U_W^n$ appear in Bob’s and Willie’s observations. Bob removes or tracks $U_B^n$ before decoding the covert message. Willie uses his side information $S$ to form the block-level estimate $\hat U_W^n=g(W^n,S)$ and tests on the innovation residual $R^n=W^n-\hat U_W^n$. Covertness is imposed on the residual laws through the relative-entropy bound $D\!\left(P_{R^n|H_1,S=s}\,\Vert\,P_{R^n|H_0,S=s}\right)\le\delta$ for all $s\in\mathcal S_{\mathrm{adm}}$. The notation is defined formally in Sections~\ref{sec:Model}--\ref{sec:constraint}.}
\label{fig:system}
\end{figure*}

\begin{table*}
\caption{Mathematical notation used throughout the paper.}
\label{tab:notations}
\setlength{\tabcolsep}{3pt}
\begin{tabular}{|p{58pt}|p{187pt}|p{58pt}|p{187pt}|}
\hline
\textbf{Symbol} & \textbf{Meaning / definition} & \textbf{Symbol} & \textbf{Meaning / definition} \\
\hline
$X_{c,t}$, $X_c^n$ & Alice's covert symbol at time $t$ and vector $X_c^n=(X_{c,1},\ldots,X_{c,n})$. & $U_{B,t}$, $U_{W,t}$ & Aggregate public components at Bob and Willie at time $t$. \\
\hline
$U_B^n$, $U_W^n$ & Length-$n$ aggregate public vectors at Bob and Willie; may arise from Alice, third-party public sources, or other trackable waveforms. & $Y_t$, $W_t$ & Bob's and Willie's observations: $Y_t=X_{c,t}+U_{B,t}+Z_{B,t}$, $W_t=X_{c,t}+U_{W,t}+Z_{W,t}$. \\
\hline
$W_{i,t}$, $W_i^n$, $i\in\{0,1\}$ & Willie's scalar/vector observation under $H_i$; $W_{0,t}$ is under no covert transmission and $W_{1,t}$ under covert transmission. & $Y_{c,t}$ & Bob's covert-domain observation after public-component removal: $Y_{c,t}=X_{c,t}+Z_{B,t}$. \\
\hline
$Z_{B,t}$, $Z_{W,t}$ & AWGN at Bob and Willie, i.i.d.\ $\mathcal N(0,\sigma_B^2)$ and $\mathcal N(0,\sigma_W^2)$. & $Z_{W,i,t}$, $i\in\{0,1\}$ & Willie's AWGN sample under $H_i$; each has marginal law $\mathcal N(0,\sigma_W^2)$. \\
\hline
$\sigma_B^2$, $\sigma_W^2$ & Noise variances at Bob and Willie. & $S$ & Willie's side information per block: all physically obtainable protocol, timing, channel/noise, calibration, and public-component information. \\
\hline
$g(\cdot)$ & Willie-side public-component canceller: $\widehat U_W^n=g(W^n,S)$. & $\widehat U_W^n$, $\widehat U_{W,i}^n$ & Willie's estimate of the aggregate public component; under $H_i$, $\widehat U_{W,i}^n=g(W_i^n,S)$. \\
\hline
$R_i^n$, $i\in\{0,1\}$ & Residual vector under $H_i$: $R_i^n=W_i^n-\widehat U_{W,i}^n$; $R^n$ denotes the generic residual statistic. & $E_{W,i,t}$ & Willie-side subtraction error: $E_{W,i,t}=U_{W,t}-\widehat U_{W,i,t}$, $i\in\{0,1\}$. \\
\hline
$H_0$, $H_1$ & No-covert and covert-transmission hypotheses. Willie tests between $P_{R_0^n\mid S=s}$ and $P_{R_1^n\mid S=s}$. & $\mathcal S_{\mathrm{adm}}$ & Admissible side-information states consistent with the SRL/innovation assumptions. \\
\hline
$P_{R_i^n\mid S=s}$, $P_{R^n\mid H_i,S=s}$ & Conditional residual law under $H_i$, given $S=s$; the two notations are used interchangeably. & $P_{W_i^n,S}$, $P_{R_i^n,S}$ & Joint laws of Willie's pre-test observation or residual with side information under $H_i$. \\
\hline
$T(w^n,s)$ & Data-processing map: $T(w^n,s)=(w^n-g(w^n,s),s)$; its residual component is $R_i^n$. & $N^n(s)$, $N_t(s)$, $I_n$ & Innovation noise vector, its $t$-th coordinate, and the $n\times n$ identity matrix. \\
\hline
$\sigma_R^2(s)$ & Innovation-domain null variance: $\mathrm{Var}(R_t\mid H_0,S=s)$. & $\sigma_{R,\min}^2$, $\sigma_{R,\max}^2$ & Bounds on the known conditioned innovation-variance profile $\sigma_R^2(s)$ in (A3). \\
\hline
$\sigma_e^2(s)$ & State-dependent irreducible cancellation-residue variance, with $\sigma_R^2(s)=\sigma_W^2+\sigma_e^2(s)$. & $\sigma_e^2$ & Stationary cancellation-residue variance, i.e., $\sigma_e^2(s)\equiv\sigma_e^2$. \\
\hline
$\sigma_0^2$ & Stationary innovation-domain null variance: $\sigma_0^2=\sigma_W^2+\sigma_e^2$. & $\delta$ & Covertness budget: $\displaystyle D\!\left(P_{R^n\mid H_1,S=s}\,\middle\|\,P_{R^n\mid H_0,S=s}\right)\le \delta,\ \forall s\in\mathcal S_{\mathrm{adm}}$. \\
\hline
$D(\cdot\|\cdot)$ & Relative entropy. & $P_t$ & Covert per-use variance: $P_t=\mathrm{Var}(X_{c,t})$. \\
\hline
$\bar P_n$ & Average covert power: $\displaystyle \bar P_n \triangleq \frac{1}{n}\sum_{t=1}^{n}\mathbb{E}\!\left[(X_{c,t})^2\right]$. & $P_n$ & Uniform covert power; exact $P_n=\sigma_0^2\rho_n$, first-order $P_n=2\sigma_0^2\sqrt{\delta/n}$. \\
\hline
$P$ & Generic uniform per-use covert variance in single-letter formulas. & $\rho_t$ & Normalized covert power: $\rho_t=P_t/\sigma_0^2$ or $\rho_t=P_t/\sigma_{0,t}^2$ in the heterogeneous case. \\
\hline
$\rho$ & Uniform normalized covert power: $\rho=P/\sigma_0^2$ when $P_t\equiv P$. & $f(\rho)$ & Relative-entropy increment: $f(\rho)=\frac12\bigl(\rho-\ln(1+\rho)\bigr)$. \\
\hline
$u_\delta$ & Unique solution of $f(u_\delta)=\delta/n$. & $\rho_n$ & Uniform equalizer solution of $\frac{n}{2}\bigl(\rho_n-\ln(1+\rho_n)\bigr)=\delta$. \\
\hline
$P_e^{(n)}$ & Bob's average block error probability. & $B_n$ & Achievable covert message size in bits over blocklength $n$. \\
\hline
$B_n^\star$ & Maximal achievable covert message size in bits over blocklength $n$. & $R_C^\star$ & Optimal SRL constant: $\displaystyle R_C^\star=\frac{\sigma_0^2}{\sigma_B^2\ln 2}\sqrt{\delta}$. \\
\hline
$\bar{R}_C$ & Baseline SRL constant: $\displaystyle \bar{R}_C=\frac{\sigma_W^2}{\sigma_B^2\ln 2}\sqrt{\delta}$. & $I(\cdot;\cdot)$ & Mutual information. \\
\hline
$h(\cdot)$ & Differential entropy. & $\Sigma_0$ & Residual covariance under $H_0$ in the colored-residual extension. \\
\hline
$\mathrm{diag}(P_t)$ & Diagonal matrix with $t$-th diagonal entry $P_t$. & $\sigma_{0,t}^2$ & Time-varying innovation-domain null variance. \\
\hline
$P_t^\star$ & First-order optimal heterogeneous scheduler, satisfying $P_t^\star\propto \sigma_{0,t}^4$. & $\ln 2$ & Natural-log constant used for rates in bits. \\
\hline
$o(1)$, $O(\cdot)$ & Landau symbols for asymptotics. & $\mathcal N(0,\cdot)$ & Zero-mean Gaussian distribution with the indicated variance/covariance. \\
\hline
$p_t$, $q_t$ & Per-coordinate marginal laws used in the \hyperref[app:KL]{Appendix} for product-additivity of relative entropy. & $ $ & $ $ \\
\hline
\end{tabular}
\end{table*}

\subsection{System Model}\label{sec:Model}

For blocklength $n$, Alice's low-power covert component is denoted by $X_c^n \triangleq (X_{c,1},\ldots,X_{c,n})$. Under no covert transmission, we set $X_c^n \equiv 0$. The public component is represented at Bob's and
Willie's receivers by 
\begin{equation}\label{eq:public}
U_B^n \triangleq (U_{B,1},\ldots,U_{B,n}), \qquad
U_W^n \triangleq (U_{W,1},\ldots,U_{W,n}).
\end{equation}
These components may arise from an Alice-generated overt layer, one or more third-party public transmissions, or other trackable public waveforms. They are not counted as covert payload and are not modeled as secrecy-bearing message variables.

Bob receives the length-$n$ block $Y^n \triangleq (Y_1,\ldots,Y_n)$, whose $t$-th component is
\begin{equation}\label{eq:Bob}
Y_t = X_{c,t}+U_{B,t}+Z_{B,t}, \qquad t=1,\ldots,n.
\end{equation}
Willie observes the length-$n$ block $W^n \triangleq (W_1,\ldots,W_n)$, whose $t$-th component is
\begin{equation}\label{eq:Willie}
W_t = X_{c,t}+U_{W,t}+Z_{W,t}, \qquad t=1,\ldots,n.
\end{equation}
Here $Z_B^n \triangleq (Z_{B,1},\ldots,Z_{B,n})$ and $Z_W^n \triangleq (Z_{W,1},\ldots,Z_{W,n})$ are Bob’s and Willie’s receiver-noise sequences, respectively. The samples $\{Z_{B,t}\}_{t=1}^n$ are i.i.d.~Gaussian random variables with distribution $\mathcal{N}(0,\sigma_B^2)$, and the samples $\{Z_{W,t}\}_{t=1}^n$ are i.i.d.~Gaussian random variables with distribution $\mathcal{N}(0,\sigma_W^2)$. The two noise sequences are mutually independent and are independent of the covert symbols and the public components. The parameters $\sigma_B^2$ and $\sigma_W^2$ denote Bob’s and Willie’s receiver-noise variances.

\paragraph{Willie's side information $S$}
Let \(S\) denote all side information physically obtainable by Willie per block. This includes the public protocol, the covert signaling law and codebook unless protected by a separately specified secret key, pilots/preambles, metadata, timing, and any calibration, channel, or noise estimates derivable from the public protocol and Willie’s received samples. The transmitted message and any private random seed or secret key, if present, are not included in \(S\).

\paragraph{Warden-favorable cancellation and innovation residual}
The two hypotheses correspond to the absence or presence of the covert component in Willie’s observation. More precisely, let
\begin{equation}\label{eq:willie_hypotheses}
\begin{aligned}
H_0 &: \quad W_{0,t}=U_{W,t}+Z_{W,0,t},\\
H_1 &: \quad W_{1,t}=X_{c,t}+U_{W,t}+Z_{W,1,t},
\end{aligned}
\qquad t=1,\ldots,n .
\end{equation}
For \(i\in\{0,1\}\), we write $W_i^n \triangleq (W_{i,1},\ldots,W_{i,n})$ for Willie’s length-\(n\) observation vector under hypothesis \(H_i\).
Willie does not observe $X_c^n$ directly. Instead, under each hypothesis, he first performs warden-favorable cancellation of the overt component using his received block and side information. The same cancellation rule is applied under both hypotheses:
\begin{equation}\label{eq:cancellation}
\widehat U_{W,i}^n = g(W_i^n,S), \qquad i\in\{0,1\},
\end{equation}
where \(g\) is Willie’s public-component canceller based on \(W^n\) and \(S\), as formalized in Assumption (A2) below. The corresponding post-cancellation residual vectors are
\begin{equation}\label{eq:residual}
R_i^n \triangleq W_i^n-\widehat U_{W,i}^n, \qquad i\in\{0,1\}.
\end{equation}
When no hypothesis index is needed, we write \(R^n\) for the generic length-\(n\) residual statistic used by Willie.

Define the Willie-side public-component estimation error under hypothesis $H_i$ as
\begin{equation}\label{eq:estimation_error}
E_{W,i,t}\triangleq U_{W,t}-\widehat U_{W,i,t}, \qquad i\in\{0,1\}.
\end{equation}
Then, under $H_0$,
\begin{equation}\label{eq:estimation_0}
R_{0,t}=E_{W,0,t}+Z_{W,0,t},
\end{equation}
whereas under $H_1$,
\begin{equation}\label{eq:estimation_1}
R_{1,t}=X_{c,t}+E_{W,1,t}+Z_{W,1,t}.
\end{equation}
Therefore, Willie’s operational statistical test is a test between the induced residual laws
\begin{equation}\label{eq:residual_test}
H_0:\ R^n\sim P_{R_0^n|S=s}
\qquad\text{vs.}\qquad
H_1:\ R^n\sim P_{R_1^n|S=s}.
\end{equation}
Equivalently, throughout the paper we write $P_{R^n|H_i,S=s}=P_{R_i^n|S=s}$ for $i\in\{0,1\}$.

\paragraph{Innovation-domain (SRL) regime assumptions}
The analysis below is intended in the SRL/innovation regime characterized by the following restrictions (stated here and used throughout):
\begin{itemize}
\item[(A1)] \emph{Trackable public component:} the aggregate component $U_W^n$ is either standardized, publicly specified, or otherwise trackable by Willie from protocol information and Willie-available estimation procedures, so that the timing/gain/state needed for subtraction is physically obtainable.
\item[(A2)] \emph{Warden-favorable cancellation:} Willie chooses the public-component canceller \(g\) from the class of physically admissible cancellation rules based on \(W^n\) and \(S\), using all physically obtainable information to reduce the residual uncertainty of the public component before his residual-domain test. The SRL analysis is conditional on the residual law induced by this chosen canceller.
\item[(A3)] \emph{Known Gaussian innovation residual after conditioning:} Conditioned on \(S=s\) under \(H_0\), the post-cancellation innovation residual is zero-mean Gaussian with a known variance profile. Formally, for every admissible side-information state \(s\), after Willie applies the chosen public-component canceller \(g\), there exists an innovation noise vector \(N^n(s)\), independent of the covert codeword \(X_c^n\), such that $R_0^n = N^n(s), \qquad R_1^n = X_c^n + N^n(s)$. In the stationary case, $N^n(s) \sim \mathcal{N}(0,\sigma_0^2 I_n),\qquad\sigma_0^2=\sigma_W^2+\sigma_e^2$, where \(I_n\) is the \(n\times n\) identity matrix. In the heterogeneous case, the coordinates are independent with $N_t(s) \sim \mathcal{N}(0,\sigma_{0,t}^2)$. Equivalently, 
\begin{equation}
    \operatorname{Var}\!\left(R_t \mid H_0, S=s\right)
    =
    \sigma_R^2(s).
\end{equation}

For regularity, assume
\begin{equation}
    0
    <
    \sigma_{R,\min}^2
    \leq
    \sigma_R^2(s)
    \leq
    \sigma_{R,\max}^2
    <
    \infty
\end{equation}
for known constants. These bounds are not an unknown fixed uncertainty interval; they only ensure that the known conditioned residual-variance profile is finite and bounded away from zero. The stationary case $\sigma_R^2(s)\equiv\sigma_0^2$ is treated as the default in Sections \ref{sec:kl-equalizer}-\ref{sec:main}.

The Gaussian innovation model in (A3) is an achievability model for calibrated cancellation regimes, not a universal model for all cancellers. It is exact in standard linear-Gaussian estimation, where the innovation sequence is zero-mean Gaussian and white, and its covariance is determined by the filtering model \cite[Ch.~5, Secs.~5.3--5.4]{anderson1979optimal}. It is also a standard analytical approximation in interference-cancellation analysis: successive interference cancellation for code-division multiple access (CDMA) systems has been analyzed using Gaussian approximations for the post-cancellation decision noise \cite{patel1994sic}; nonlinear minimum mean-square-error (MMSE) and probabilistic data association (PDA) multiuser detectors approximate residual multiple-access interference as multivariate Gaussian \cite{tan2006gaussian}; and full-duplex models represent residual self-interference after cancellation as an additive Gaussian impairment with variance set by the residual self-interference level \cite{zlatanov2017residual}. Thus, (A3) applies when Willie’s conditioned post-cancellation null residual can be modeled as zero-mean Gaussian with known variance. If this model is rejected by measurements or residual diagnostics, the achievability result in Section~V does not apply, although Lemma~\ref{lem:Gaussian} still shows that Gaussian residuals are least detectable for a fixed coordinate-variance profile.

\item[(A4)] \emph{No hidden-resource public-component design:} the public component is not used as a hiding resource unknown to Willie. In particular, Alice cannot encode hidden resources into any Alice-controlled overt layer, scheduling pattern, silent-gap structure, or protocol feature that Willie cannot track or condition out through $S$.
\end{itemize}

Not every cancellation map \(g(W^n,S)\) yields the known Gaussian innovation residual in (A3). The results in Sections~\ref{sec:kl-equalizer}--\ref{sec:main} apply only when the residual law induced by Willie’s chosen canceller satisfies (A3). If cancellation removes part of \(X_c^n\), induces dependence between \(X_c^n\) and the innovation residual, produces a non-Gaussian residual, or leaves non-negligible residual-variance uncertainty, then the present SRL theorem does not apply. In such cases, including untracked fading, interference, or power uncertainty, noise/channel-uncertainty, jammer- or artificial-noise-assisted covertness, or effective-secrecy mechanisms may instead dominate \cite{Bash2017,Bash2018,He2017,Shahzad2017,Sobers2015,Hou2014,Hou2017}.

\begin{remark}[Origin and multiplicity of public components] \label{rem:public}
The analysis depends on Willie’s final post-cancellation residual, not on the origin or number of public components. Thus, \(U_W^n\) may comprise an Alice-generated overt signal, third-party public waveforms, or any trackable combination thereof. If $U_{W,t}=\sum_{k=1}^{K}U_{W,k,t}$, Willie may cancel the components jointly or successively using his available side information. The results in Sections~\ref{sec:kl-equalizer}--\ref{sec:main} remain valid whenever the resulting aggregate residual satisfies (A3).
\end{remark}

\paragraph{Bob-side overt removability}
For the achievability analysis, only the covert component $X_c^n$ carries the covert payload. Bob is assumed to estimate, track, decode, or otherwise remove the aggregate public component $U_B^n$ affecting his observation by using the public protocol specification, his own received signal, and standard receiver-side estimation procedures. This assumption does not require Bob to possess hidden side information about the public component. It only means that, after Bob’s receiver processing, the residual public-component error is either negligible on the $\sqrt n$ scale or can be absorbed into an effective post-removal noise variance. Consequently, Bob’s effective observation for the counted covert payload is the covert-only AWGN observation used in Section~\ref{sec:Achievability}. If a non-negligible independent Bob-side residual remains, the same first-order expression is obtained by replacing $\sigma_B^2$ with the corresponding effective post-removal noise-plus-residual variance.

The overall warden-favorable innovation-regime system model is summarized in Figure~\ref{fig:system}.

\subsection{Covertness Constraint}\label{sec:constraint}
Fix a budget $\delta>0$, and covertness is enforced on the innovation residual and conditioned on Willie's side information. Let $\mathcal S_{\mathrm{adm}}$ denote the set of side-information states consistent with (A1)--(A4), and let $P_{R_i^n\mid S=s}$ denote the conditional law of Willie's residual vector under $H_i$, $i\in\{0,1\}$. For probability measures $P$ and $Q$, the relative entropy from $P$ to $Q$ is denoted by $D(P\|Q)$ and is defined as 
\begin{equation}\label{eq:relative_entropy_def} 
D(P\|Q)
\triangleq
\int \log\!\left(\frac{\mathrm dP}{\mathrm dQ}\right)\,\mathrm dP,
\end{equation}
when $P$ is absolutely continuous with respect to $Q$, and $D(P\|Q)=+\infty$ otherwise. 
As in prior SRL covert-communication works \cite{Bash2012,Bash2013,Wang2016,Bloch2016,Tahmasbi2019}, we use relative entropy to quantify covertness because it controls the statistical distinguishability between Willie’s residual laws under $H_0$ and $H_1$ and therefore gives a detector-agnostic constraint on Willie’s ability to decide whether covert transmission is present. The conditioned post-subtraction relative-entropy constraint is
\begin{equation}\label{eq:KL-post}
D\!\left(P_{R_1^n\mid S=s}\,\big\|\,P_{R_0^n\mid S=s}\right)\le \delta,\qquad \forall s\in \mathcal S_{\mathrm{adm}}.
\end{equation}

\subsection{Design Variables and Objective}\label{sec:design}
We optimize over all blocklength-$n$ encoder--decoder pairs for the covert layer (no structural restrictions). The covert symbols are taken without loss of generality to be zero-mean with per-use variances $P_t=\operatorname{Var}(X_{c,t})$, and we operate in the vanishing-power regime
\begin{equation}
    \max_{1\le t\le n}\frac{P_t}{\sigma^2_{R,\min}}\to 0
    \qquad \text{as } n\to\infty .
\end{equation}
The corresponding block-average covert power is $\bar P_n \triangleq \frac{1}{n}\sum_{t=1}^n\mathbb E[X_{c,t}^2]$, where $\mathbb E[\cdot]$ denotes expectation. The condition above implies $\bar P_n\to 0$. In the uniform case $P_t\equiv P_n$, one has $\bar P_n=P_n$.

Let $B_n$ denote the number of message bits of a given code, and $B_n^\star$ denote the maximal achievable value under \eqref{eq:KL-post} among codes with average block error $P_e^{(n)} \to 0$ as $n \to \infty$. The problem is
\begin{equation}\label{eq:problem}
B_n^\star=\max B_n \quad \text{subject to \eqref{eq:KL-post}.}
\end{equation}
All logarithms are natural; rates are reported in bits via division by $\ln 2$.

\section{Relative-Entropy Control and Power Design}\label{sec:kl-equalizer}

This section turns the innovation-domain covertness constraint \eqref{eq:KL-post} into covert-power constraints and corresponding power-allocation schedules. The expressions derived below depend on the covertness budget $\delta$ and on the innovation-domain null variance. In the stationary innovation case, $\sigma_R^2(s)\equiv \sigma_0^2$ for all admissible $s\in\mathcal S_{\mathrm{adm}}$, and the relative-entropy expressions below hold uniformly over $s\in\mathcal S_{\mathrm{adm}}$ when conditioned on $S=s$.

Proofs of the Gaussian relative-entropy identity, the increment bounds, and the product-additivity formula used in \eqref{eq:KLconvert}--\eqref{eq:KL-sum} are given in the \hyperref[app:KL]{Appendix}.

By Assumption (A3), the stationary case considered in Sections~\ref{sec:kl-equalizer}--\ref{sec:main} has a conditioned product-Gaussian null residual. In the Gaussian signaling calculations of this section, independent symbols $X_{c,t}\sim\mathcal N(0,P_t)$ therefore induce $R_{0,t}\mid S=s\sim \mathcal N(0,\sigma_0^2),\qquad R_{1,t}\mid S=s\sim \mathcal N(0,\sigma_0^2+P_t)$, independently across \(t\). We call \(P_t\equiv P_n\) a uniform schedule and \(\{P_t\}_{t=1}^n\) a non-uniform schedule.

\subsection{Relative Entropy Under a Variance Increase (Uniform Schedules)}\label{sec:uniform}
Let $P\ge 0$ denote the per-use covert variance in the uniform schedule, so that $P_t\equiv P$. In achievability, this is enforced by choosing $X_{c,t}\sim \mathcal N(0,P)$ independent of the null innovation term $N_{t}(s)$, where $R_{0,t}=N_{t}(s)$. Set $\rho \triangleq P/\sigma_0^2$. Then
\begin{equation}\label{eq:KLconvert}
D\!\left(\mathcal N(0,\sigma_0^2+P)\,\middle\|\,\mathcal N(0,\sigma_0^2)\right)
=\frac12\bigl(\rho-\ln(1+\rho)\bigr).
\end{equation}
Moreover, for all $\rho \ge 0$,
\begin{equation}\label{eq:KLallt}
\frac{\rho^2}{4(1+\rho)} \le \frac12\bigl(\rho-\ln(1+\rho)\bigr) \le \frac{\rho^2}{4}.
\end{equation}

\subsection{Relative Entropy Under a Variance Increase (Non-Uniform Schedules)}\label{sec:non_uniform}
With per-use variances $\{P_t\}$, set \(\rho_t \triangleq P_t/\sigma_0^2\). If $R_0^n$ and $R_1^n$ have product laws conditioned on $S=s$, then
\begin{equation}\label{eq:KL-sum}
D\!\left(P_{{R_1^n}|S=s}\,\middle\|\,P_{{R_0^n}|S=s}\right)
= \sum_{t=1}^n \frac12\bigl(\rho_t-\ln(1+\rho_t)\bigr).
\end{equation}

Under the same conditioned product-law residual model used for \eqref{eq:KL-sum}, namely, that conditioned on \(S=s\), the residual samples are independent across time indices \(t=1,\ldots,n\) under both \(H_0\) and \(H_1\), the lower bound in \eqref{eq:KLallt} yields
\begin{equation}\label{eq:KL-lower-quad}
D\!\left(P_{{R_1^n}|S=s}\,\middle\|\,P_{{R_0^n}|S=s}\right)
\ge \frac14\sum_{t=1}^n \frac{\rho_t^2}{1+\rho_t}.
\end{equation}
If \(\max_{1\le t\le n}\rho_t \to 0\) as \(n\to\infty\), then
\begin{equation}
\min_{1\le t\le n}\frac{1}{1+\rho_t}=1-o(1),
\end{equation}
and therefore
\begin{equation}\label{eq:KL-lower-quad1}
D \ge \frac{1-o(1)}{4}\sum_{t=1}^n \rho_t^2
= \frac{1-o(1)}{4\sigma_0^4}\sum_{t=1}^n P_t^2.
\end{equation}
Consequently, \eqref{eq:KL-lower-quad}--\eqref{eq:KL-lower-quad1} are used in this paper only under the conditioned product-law residual model, i.e., when, for each admissible \(s\), the residual samples are independent across time indices under both hypotheses after conditioning on \(S=s\).

\subsection{Total-Power Optimality Under a Relative-Entropy Budget}
\begin{lemma}[Uniform schedule maximizes total power for a given relative-entropy budget] \label{lem:uniform-opt}
Let $f(\rho) = \frac12(\rho - \ln(1+\rho)), \rho \ge 0$. For any $\{\rho_t\}_{t=1}^n$ with $\sum_{t=1}^n f(\rho_t) \le \delta$, $ \sum_{t=1}^n \rho_t \le n u_\delta$, where $u_\delta$ solves $f(u_\delta) = \delta/n$. Equality holds if and only if $\rho_t \equiv u_\delta$.
\end{lemma}

\begin{proof}
Since $f''(\rho)=\frac{1}{2(1+\rho)^2}>0$ $(\rho\ge 0)$, the function $f$ is strictly convex on $[0,\infty)$. Let
\[
\bar\rho \triangleq \frac1n\sum_{t=1}^n \rho_t.
\]
By Jensen’s inequality,
\[
f(\bar\rho)\le \frac1n\sum_{t=1}^n f(\rho_t)\le \frac{\delta}{n}.
\]
Also,
\[
f'(\rho)=\frac{\rho}{2(1+\rho)}\ge 0 \qquad (\rho\ge 0),
\]
with strict inequality for $\rho>0$, so $f$ is increasing on $[0,\infty)$. Hence $\bar\rho \le u_\delta$, where $u_\delta$ is the unique solution of $f(u_\delta)=\delta/n$. Therefore,
\[
\sum_{t=1}^n \rho_t = n\bar\rho \le n u_\delta.
\]

If equality holds, then necessarily $\bar\rho=u_\delta$, hence
$f(\bar\rho)=\delta/n$, so equality must hold in Jensen’s inequality. Since $f$
is strictly convex, equality in Jensen occurs if and only if
$\rho_t\equiv \bar\rho$ for all $t$, i.e., $\rho_t\equiv u_\delta$. The converse follows.
\end{proof}

\subsection{Exact Equalizer (Uniform Schedule)}
\begin{lemma}[Equalizer exists, is unique, and meets the relative-entropy budget]\label{lem:equalizer}
When $P_t \equiv P_n$, define $\rho_n \triangleq P_n/\sigma_0^2$. For any $\delta > 0$ there is a unique $\rho_n > 0$ solving
\begin{equation}\label{eq:equalizer}
\frac{n}{2}\bigl(\rho_n - \ln(1+\rho_n)\bigr) = \delta,
\end{equation}
and the uniform choice $P_n = \sigma_0^2 \rho_n$ meets \eqref{eq:KL-post}  with equality (uniformly over $s \in S_{\mathrm{adm}}$ in the stationary case). Moreover, $\rho_n = 2\sqrt{\delta/n} + O(\delta/n)$.
\end{lemma}

\begin{proof}
Define $g(\rho)\triangleq \frac{n}{2}\bigl(\rho-\ln(1+\rho)\bigr),  \rho\ge 0$. Then
\[
\begin{aligned}
g(0) &= 0,\\[2pt]
g'(\rho) &= \frac{n}{2}\frac{\rho}{1+\rho}\ge 0,\\[2pt]
g''(\rho) &= \frac{n}{2}\frac{1}{(1+\rho)^2}>0.
\end{aligned}
\]
Hence $g$ is continuous and strictly increasing on $[0,\infty)$, with $g(\rho)\to\infty$ as $\rho\to\infty$. Therefore, by the intermediate value theorem, for every $\delta>0$ there exists a unique $\rho_n>0$ such that $g(\rho_n)=\delta$, which is exactly \eqref{eq:equalizer}.

Under the conditioned i.i.d.\ residual model in the stationary case, the
uniform schedule $P_t\equiv P_n$ gives
\[
D\!\left(
P_{R_1^n\mid S=s}
\,\middle\|\,
P_{R_0^n\mid S=s}
\right)
=
\frac{n}{2}\bigl(\rho_n-\ln(1+\rho_n)\bigr).
\]
Thus, choosing $P_n=\sigma_0^2\rho_n$ with $\rho_n$ satisfying \eqref{eq:equalizer} meets \eqref{eq:KL-post} with equality.

For the asymptotics, since $\delta$ is fixed and $g(\rho_n)=\delta$, one has
$\rho_n\to 0$ as $n\to\infty$. Using $\ln(1+\rho)=\rho-\frac{\rho^2}{2}+O(\rho^3), \rho\to 0$, we obtain $\delta=\frac{n}{2}\bigl(\rho_n-\ln(1+\rho_n)\bigr)=\frac{n}{4}\rho_n^2+O(n\rho_n^3)$.

Hence $\rho_n=O(n^{-1/2})$, and substituting this back into the remainder term
yields $\delta = \frac{n}{4}\rho_n^2 + O(n^{-1/2})$, which implies $\rho_n^2 = \frac{4\delta}{n}+O(n^{-3/2})$.

Taking square roots gives $\rho_n = 2\sqrt{\delta/n}+O(\delta/n)$.
\end{proof}

By Lemma~\ref{lem:uniform-opt}, the uniform $\rho_n$ also maximizes the total covert power $\sum_{t=1}^n P_t$ among all schedules with the same relative-entropy budget.

\subsection{Closed-Form Design}
Using $\rho_n = 2\sqrt{\delta/n} + O(\delta/n)$ in $P_n = \sigma_0^2 \rho_n$ gives the closed-form first-order uniform schedule
\begin{equation}\label{eq:achi}
P_n = 2\sigma_0^2\sqrt{\delta/n}.
\end{equation}
This closed-form choice has the same first-order expansion as the exact equalizer \(P_n=\sigma_0^2\rho_n\). This design rule is used in Section~\ref{sec:Achievability}.

\subsection{Heterogeneous Innovation Variances: The $\sigma^4$ Scheduler} \label{sec:scheduler}

We now relax the stationary innovation assumption and allow the conditioned innovation-domain null variance to vary across channel uses. Specifically, conditioned on $S=s$, let $R_{0,t}\mid S=s \sim \mathcal N(0,\sigma_{0,t}^2)$, and, with independent Gaussian covert symbols $X_{c,t}\sim\mathcal N(0,P_t)$, let  $R_{1,t}\mid S=s \sim \mathcal N(0,\sigma_{0,t}^2+P_t)$.

Set $\rho_t \triangleq P_t/\sigma_{0,t}^2$. Assume \(\max_t \rho_t \to 0\). Under the conditioned product-law residual model, namely, when the induced residual samples are independent across \(t\) after conditioning on \(S=s\), the conditioned block relative entropy is
\begin{equation}
D\!\left(
P_{R_1^n\mid S=s}
\,\middle\|\,
P_{R_0^n\mid S=s}
\right)
=
\sum_{t=1}^n \frac12\left(\rho_t-\ln(1+\rho_t)\right).
\end{equation}

Using the small-signal expansion uniformly in \(t\),
\begin{equation}
D\!\left(
P_{R_1^n\mid S=s}
\,\middle\|\,
P_{R_0^n\mid S=s}
\right)
=
\frac{1+o(1)}{4}
\sum_{t=1}^n \frac{P_t^2}{\sigma_{0,t}^4}.
\end{equation}

Hence, to first order, the conditioned relative-entropy budget \(D\le \delta\) induces the ellipsoidal constraint
\begin{equation}
\sum_{t=1}^n \frac{P_t^2}{4\sigma_{0,t}^4}\le \delta.
\end{equation}

To maximize Bob’s first-order throughput, it is enough to maximize \(\sum_{t=1}^n P_t\) subject to the above constraint. By the Cauchy--Schwarz inequality,
\begin{equation}
\begin{aligned}
\left(\sum_{t=1}^n P_t\right)^2
&= \left(\sum_{t=1}^n \frac{P_t}{2\sigma_{0,t}^2}\, 2\sigma_{0,t}^2\right)^2 \\
&\le \left(\sum_{t=1}^n \frac{P_t^2}{4\sigma_{0,t}^4}\right)
\left(4\sum_{t=1}^n \sigma_{0,t}^4\right) \\
&\le 4\delta \sum_{t=1}^n \sigma_{0,t}^4.
\end{aligned}
\label{eq:heterogeneous-cs-bound}
\end{equation}

The first inequality in \eqref{eq:heterogeneous-cs-bound} is tight if and only if the two sequences
\[
\left\{\frac{P_t}{2\sigma_{0,t}^2}\right\}_{t=1}^n
\quad\text{and}\quad
\left\{2\sigma_{0,t}^2\right\}_{t=1}^n
\]
are linearly dependent. Equivalently, there exists a constant \(c\ge 0\) such that $P_t = c\,\sigma_{0,t}^4, t=1,\ldots,n$. Since the first-order objective \(\sum_{t=1}^n P_t\) is increasing in \(c\), the optimal schedule uses the relative-entropy budget with equality. Hence
\[
\sum_{t=1}^n \frac{c^2\sigma_{0,t}^8}{4\sigma_{0,t}^4}
=
\frac{c^2}{4}\sum_{t=1}^n \sigma_{0,t}^4
=
\delta,
\]
which gives
\[
c=
\frac{2\sqrt{\delta}}
{\sqrt{\sum_{j=1}^n\sigma_{0,j}^4}}.
\]
Therefore the first-order optimal schedule is
\begin{equation}
P_t^\star
=
\frac{2\sqrt{\delta}\,\sigma_{0,t}^4}
{\sqrt{\sum_{j=1}^n \sigma_{0,j}^4}},
\end{equation}
and the corresponding first-order total covert power is
\begin{equation}
\sum_{t=1}^n P_t^\star
=
2\sqrt{\delta}\,
\sqrt{\sum_{t=1}^n \sigma_{0,t}^4}.
\end{equation}

At Bob, under \(\max_t P_t/\sigma_B^2 \to 0\), the first-order upper bound is
\begin{equation}
B_n \le
\frac{1}{2 \ln 2}
\sum_{t=1}^n
\ln\left(1+\frac{P_t}{\sigma_B^2}\right)
\le
\frac{1+o(1)}{2\sigma_B^2 \ln 2}\sum_{t=1}^n P_t.
\end{equation}
Thus, under the optimal schedule,
\begin{equation}
B_n
=
\frac{1}{\sigma_B^2\ln 2}
\sqrt{\delta\sum_{t=1}^n \sigma_{0,t}^4}\,(1+o(1)).
\end{equation}

Thus, in the heterogeneous innovation case, the conditioned relative-entropy budget yields a first-order optimal \(\sigma^4\)-type scheduler \(P_t^\star \propto \sigma_{0,t}^4\), and, under the same conditioned product-law residual model used above, the same Cauchy–Schwarz argument gives the matching first-order converse.

\subsection{Pre-Residualization Conservative Design (Optional Robustness)}\label{sec:residualization}
As a robustness alternative, one may impose covertness before residualization on Willie’s full pre-test information. Unlike the main relative-entropy constraint in \eqref{eq:KL-post}, the present robustness alternative is imposed on the joint law of Willie’s full pre-test information $(W^n,S)$ under each hypothesis:
\begin{equation}
D\!\left(P_{W_1^n,S}\,\middle\|\,P_{W_0^n,S}\right)
\le \delta.
\end{equation}
Since $(R^n,S)$ is a deterministic function of $(W^n,S)$ via $T(w^n,s) \triangleq \bigl(w^n-g(w^n,s),\, s\bigr)$, data processing gives
\begin{equation}
D\!\left(P_{W_1^n,S}\,\middle\|\,P_{W_0^n,S}\right)
\ge
D\!\left(P_{R_1^n,S}\,\middle\|\,P_{R_0^n,S}\right).
\end{equation}
Hence the pre-test constraint is conservative relative to the innovation-residual constraint in \eqref{eq:KL-post}, when both are based on the same Willie-available information. In the stationary Gaussian baseline, replacing $\sigma_0^2$ by $\sigma_W^2$ recovers the baseline SRL constant $\bar R_C \triangleq \frac{\sigma_W^2}{\sigma_B^2\ln 2}\sqrt{\delta}$, which corresponds to the case without the residual-floor boost.

\section{Interpretation of the Innovation Regime and Boundary to Linear-Law Settings}\label{sec:regimes}

This section explains how the innovation-domain null variance should be understood physically, gives representative practical examples consistent with the assumptions in Section \ref{sec:Model}, and clarifies the boundary between the SRL/innovation regime studied here and out-of-scope linear-law/effective-secrecy regimes.

\subsection{Residual Floor as Cancellation-Limited Innovation Variance} \label{sec:variance}
Under the assumptions already stated in Section \ref{sec:Model}, Willie’s innovation-domain null, conditioned on $S=s$, can be interpreted as $R_{0,t}=N_t(s), \qquad N_t(s)\sim \mathcal N(0,\sigma_R^2(s))$, where $\sigma_R^2(s)=\sigma_W^2+\sigma_e^2(s)$ and $\sigma_e^2(s)\ge 0$ represents the irreducible cancellation residue after warden-favorable subtraction. Under covert transmission, the corresponding residual model is
$R_{1,t}=X_{c,t}+N_t(s)$. In the stationary case, $\sigma_e^2(s)\equiv \sigma_e^2>0$, so $\sigma_0^2=\sigma_W^2+\sigma_e^2>\sigma_W^2$ enters all relative-entropy expressions in Section \ref{sec:kl-equalizer} and beyond.

The role of $\sigma_e^2$ is purely in the innovation domain: it quantifies the residual noise variance that remains after Willie exploits $S$ and applies the best block-level cancellation of the aggregate public component. Typical physical contributors include finite cancellation depth, estimation residue when pilot symbols are used, hardware distortion such as phase noise and nonlinearities, and quantization or front-end noise. In particular, $\sigma_e^2$ is not modeled as a designable entropy source or hidden resource for the covert encoder; it is a residual property of Willie’s own innovation-domain null after warden-favorable cancellation under the stated regime
restrictions.

\subsection{Practical Examples Consistent With (A1)--(A4)}
\label{sec:anchors}

Two representative examples illustrate settings in which (A1)--(A4) are plausible. In both cases, after conditioning on \(S=s\), Willie’s post-cancellation residual is modeled as
\begin{equation}
\begin{aligned}
R_{0,t} &= N_t(s),\\
R_{1,t} &= X_{c,t}+N_t(s),\\
N_t(s) &\sim \mathcal N(0,\sigma_R^2(s)).
\end{aligned}
\end{equation}
Hence, under \(H_0\), \(R_{0,t}\sim\mathcal N(0,\sigma_R^2(s))\), while under \(H_1\), for \(X_{c,t}\sim\mathcal N(0,P_t)\), \(R_{1,t}\sim\mathcal N(0,\sigma_R^2(s)+P_t)\).

\paragraph{Known receiver-local component} A known receiver-local component may arise from receiver operation, calibration, or a locally available reference signal. Its realization and associated calibration information are included in \(S\). After subtracting this component and applying the overt-component cancellation allowed by (A2), \(N_t(s)\) collects receiver noise and any remaining local- or overt-cancellation residue.

\paragraph{Trackable public waveform(s) + pilots} The public waveform(s) and pilots/preamble are standardized, publicly specified, or otherwise trackable. Willie uses them to estimate timing and channel state and subtract the overt component. Here, \(N_t(s)\) includes the remaining tracking, estimation, and hardware-mismatch residue.

Both examples fall within the present SRL/innovation regime only when the resulting residual satisfies the known Gaussian innovation condition in (A3); in the stationary case, \(\sigma_R^2(s)\equiv\sigma_0^2\). Otherwise, the failure mode described next applies.

\subsection{Failure Mode: Linear-Law/Effective-Secrecy Regimes (Out of Scope)}
If (A3) fails because Willie is left with non-vanishing uncertainty about the post-cancellation residual variance, or because the residual is not additive Gaussian after conditioning, then the asymptotic behavior need not be governed by the SRL covert overlay alone. Examples include deep fading that is not trackable from \(S\), unknown interference, or unknown received power. In such cases, mechanisms associated with noise uncertainty or effective secrecy may dominate, potentially leading to linear-law behavior rather than the SRL regime studied here.

This regime is conceptually closer to the effective-secrecy formulations \cite{Hou2014,Hou2017}, where stealth is formulated by requiring Willie’s observation law under communication to remain close to a default law, and where confusion and stealth are analyzed jointly through a relative-entropy criterion.

The present paper is not an effective-secrecy analysis. Equation \eqref{eq:KL-post} imposes only a stealth-style detectability constraint on Willie’s post-cancellation innovation residual under the SRL restrictions (A1)–(A4); it does not add a confusion requirement and does not model the overt layer as a secrecy-bearing message. Accordingly, Sections~\ref{sec:kl-equalizer}--\ref{sec:main} claim SRL scaling only within the innovation regime defined in Section \ref{sec:problem}, and not in linear-law/effective-secrecy settings.

\section{Achievability of the Covert Component}\label{sec:Achievability}

This section gives an achievability argument for the covert component $X_c^n$ under the SRL/innovation regime of Section~\ref{sec:regimes}. The scheme meets the conditioned post-subtraction relative-entropy budget via the equalizer power allocation \eqref{eq:achi} and achieves the first-order constant $R_C^\star$ in \eqref{eq:main-constant}.

\subsection{Covert Layer From the Relative-Entropy Budget}
In the stationary innovation case, conditioned on any $S=s\in\mathcal S_{\mathrm{adm}}$, the induced residual satisfies $R_{0,t}\mid S=s \sim \mathcal N(0,\sigma_0^2)$, while, with \(X_{c,t}\sim \mathcal N(0,P_n)\) i.i.d.,
$R_{1,t}\mid S=s \sim \mathcal N(0,\sigma_0^2+P_n)$. Generate covert symbols i.i.d.\ $\mathcal N(0,P_n)$, independent of all else. With the uniform schedule from \eqref{eq:achi}, Section \ref{sec:kl-equalizer} guarantees $D(P_{R^n_1|S=s}\|P_{R^n_0|S=s}) \le \delta,\quad \forall s \in \mathcal{S}_{\rm adm}$.

\subsection{Decoding at Bob and Reliability}
Under the Bob-side removability assumption stated in Section~\ref{sec:Model}, Bob can recover, track, estimate, or otherwise remove the aggregate public component affecting his observation, using the protocol structure together with his own receiver observations and standard receiver estimates, with a residual negligible at the $\sqrt{n}$ scaling of interest. Therefore, the throughput counted in $B_n$ is carried only by the covert component $X_c^n$, and Bob’s effective observation for the covert component is $Y_{c,t}=X_{c,t}+Z_{B,t}$, i.e., an AWGN observation with noise variance $\sigma_B^2$.  For i.i.d.\ Gaussian codebooks and maximum-likelihood (ML) decoding over the effective AWGN channel, the scheme achieves
\begin{align*}
B_n=\frac{n}{2\ln 2}\ln\!\left(1+\frac{P_n}{\sigma_B^2}\right)\,(1+o(1)).
\end{align*}
Substituting ~\eqref{eq:achi} and using $\ln(1+x)=x+O(x^2)$ as $x\to 0$ yields
\begin{align*}
B_n=\frac{\sigma_0^2}{\sigma_B^2\ln 2}\sqrt{n\,\delta}\,(1+o(1))
=R_C^\star \sqrt{n}\,(1+o(1)),
\end{align*}
matching \eqref{eq:main-constant}.

\section{Converse: No Scheme Beats $R_C^\star$}\label{sec:Converse}

We upper-bound the maximal covert bits \(B_n\) under the conditioned post-subtraction relative-entropy budget \eqref{eq:KL-post} in the stationary specialization of (A3), where Willie’s null residual is Gaussian after conditioning. This specifies the null law tested by Willie; it does not impose a Gaussian or product structure on the covert codeword. Conditioned on any admissible \(S=s\), $R_0^n|S=s=N^n(s), \qquad R_1^n|S=s=X_c^n+N^n(s)$, where \(N^n(s)\sim\mathcal{N}(0,\sigma_0^2 I_n)\) is independent of \(X_c^n\). The codeword distribution itself may be arbitrary, including sparse or on-off signaling, provided it has zero mean, per-use variances \(P_t=\operatorname{Var}(X_{c,t})\), and satisfies the vanishing-power condition. The key step is Lemma~\ref{lem:Gaussian}, which shows that, for a fixed coordinate-variance profile, the Gaussian variance-increase residual is least detectable relative to the Gaussian null law. The rest of the argument converts the relative-entropy budget into a quadratic \(\ell_2\)-constraint on \(\{P_t\}\), then uses the AWGN information bound at Bob and Cauchy--Schwarz to obtain the first-order constant \(R_C^\star\).

Let $P_t = \mathrm{Var}(X_{c,t})$ and $\rho_t = P_t/\sigma_0^2$. We work under the vanishing-power regime $\max_t \rho_t \to 0$.

\begin{lemma}[Gaussian variance increase is least detectable]\label{lem:Gaussian}
Let $Q_0=\mathcal N(0,\sigma_0^2 I_n)$ with \(\sigma_0^2>0\). Let \(R^n\) be any zero-mean random vector with finite second moments and coordinate variances $\operatorname{Var}(R_t)=\sigma_0^2+P_t,\qquad t=1,\ldots,n,$ where \(P_t\ge 0\). Define \(\rho_t=P_t/\sigma_0^2\). Then $D(P_{R^n}\|Q_0)\ge\frac12\sum_{t=1}^n\left(\rho_t-\ln(1+\rho_t)\right)$. Equality is attained when the coordinates of \(R^n\) are independent Gaussian random variables with $R_t\sim \mathcal N(0,\sigma_0^2+P_t)$.
\end{lemma}

\begin{proof}
If \(P_{R^n}\) is not absolutely continuous with respect to \(Q_0\), then \(D(P_{R^n}\|Q_0)=+\infty\), and the claim is immediate. Otherwise,
\begin{align*}
D(P_{R^n}\|Q_0)
=
-h(R^n)
+\frac n2\ln(2\pi\sigma_0^2)
+\frac{1}{2\sigma_0^2}\mathbb E\|R^n\|^2 ,
\end{align*}
where $h(\cdot)$ denotes differential entropy. 
Let \(K_R\) be the covariance matrix of \(R^n\). Since Gaussian random vectors maximize differential entropy for a fixed covariance, $h(R^n)\le\frac12\ln\big((2\pi e)^n\det K_R\big)$. By Hadamard's inequality, $\det K_R\le\prod_{t=1}^n(\sigma_0^2+P_t)$. Also, $\mathbb E\|R^n\|^2 = \sum_{t=1}^n(\sigma_0^2+P_t)$. Substituting these bounds gives
\begin{align*}
D(P_{R^n}\|Q_0)
\ge
\frac12\sum_{t=1}^n
\left(
\frac{\sigma_0^2+P_t}{\sigma_0^2}
-1
-\ln\frac{\sigma_0^2+P_t}{\sigma_0^2}
\right),
\end{align*}
which is equivalent to
\begin{align*}
D(P_{R^n}\|Q_0)
\ge
\frac12\sum_{t=1}^n
\left(\rho_t-\ln(1+\rho_t)\right).
\end{align*}
Equality holds when \(R^n\) is Gaussian with diagonal covariance \(\operatorname{diag}(\sigma_0^2+P_1,\ldots,\sigma_0^2+P_n)\).
\end{proof}

\subsection{From Relative Entropy to an $\ell_2$ Power Budget}
Under Assumption (A3), for each admissible \(s\),
\begin{align*}
P_{R_0^n|S=s}=\mathcal N(0,\sigma_0^2I_n),
\qquad
R_1^n|S=s=X_c^n+N^n(s).
\end{align*}
The distribution of \(X_c^n\) may be arbitrary subject to the second-moment constraints \(P_t=\operatorname{Var}(X_{c,t})\). Applying
Lemma \ref{lem:Gaussian} to \(R^n=R_1^n|S=s\) and \(Q_0=P_{R_0^n|S=s}\) gives
\begin{align*}
D(P_{R_1^n|S=s}\|P_{R_0^n|S=s})
\ge
\frac12\sum_{t=1}^n
\left(\rho_t-\ln(1+\rho_t)\right).
\end{align*}
Using the lower bound
\begin{align*}
\frac12(\rho_t-\ln(1+\rho_t))
\ge
\frac{\rho_t^2}{4(1+\rho_t)}
\end{align*}
and \(\max_t \rho_t\to 0\), we obtain
\begin{align*}
D(P_{R_1^n|S=s}\|P_{R_0^n|S=s})
\ge
\frac{1-o(1)}{4}\sum_{t=1}^n \rho_t^2
=
\frac{1-o(1)}{4\sigma_0^4}
\sum_{t=1}^n P_t^2 .
\end{align*}
Since the constraint \eqref{eq:KL-post} enforces $D(P_{R_1^n|S=s}\|P_{R_0^n|S=s})\le \delta$ uniformly over \(s\in S_{\rm adm}\), it follows that 
\begin{equation}\label{eq:l2-budget}
\sum_{t=1}^n P_t^2 \le 4(1+o(1))\sigma_0^4\,\delta.
\end{equation}

\subsection{Number of Covert Bits Received by Bob}
For each use, $I(X_{c,t}; X_{c,t}+Z_{B,t})\le \frac12\ln(1+P_t/\sigma_B^2)$, where \(I(\cdot;\cdot)\) denotes mutual information. Hence
\begin{align*}
B_n\le \frac{1}{2\ln 2}\sum_{t=1}^n \ln\!\left(1+\frac{P_t}{\sigma_B^2}\right)
\le \frac{1}{2\ln 2\,\sigma_B^2}\sum_{t=1}^n P_t.
\end{align*}
By the Cauchy--Schwarz inequality and \eqref{eq:l2-budget},
\begin{align*}
\sum_{t=1}^n P_t \le \sqrt{n}\left(\sum_{t=1}^n P_t^2\right)^{1/2}
\le 2\sigma_0^2 \sqrt{n\,\delta}\,(1+o(1)).
\end{align*}
Therefore,
\begin{equation}
B_n \le \frac{\sigma_0^2}{\sigma_B^2\ln 2}\sqrt{n\,\delta}\,(1+o(1))
= R_C^\star \sqrt{n}\,(1+o(1)),
\end{equation}
which is the converse bound with the constant in~\eqref{eq:main-constant}.

\vspace{0.25em}
\section{First–Order Covert Throughput (Main Result)}\label{sec:main}

This section packages the ingredients developed so far into a single statement. Section~\ref{sec:kl-equalizer} converts the conditioned post-subtraction relative-entropy budget into the equalizer power allocation \eqref{eq:achi}. Section~\ref{sec:regimes} specifies the SRL/innovation regime (A1)--(A4) under which Willie’s innovation-domain null variance is known in the stationary case and strictly larger than $\sigma_W^2$ due to an irreducible cancellation residue $\sigma_e^2>0$. Combining the achievability in Section~\ref{sec:Achievability} and the converse in Section~\ref{sec:Converse} yields the exact $\sqrt{n}$ constant below.

\begin{theorem}[Optimal covert throughput under conditioned post--subtraction relative-entropy constraint]\label{thm:main}
Let the model be as in Section \ref{sec:problem} and suppose:
\begin{enumerate}[label={(A\arabic*)},ref={(A\arabic*)}]
\item[(T1)] \textbf{Covertness (innovation-domain, conditioned):} $D(P_{R_1^n\mid S=s}\|P_{R_0^n\mid S=s})\le \delta$ for fixed $\delta>0$ and all $s\in\mathcal S_{\mathrm{adm}}$. Here, $R_0^n$ and $R_1^n$ are the residual vectors induced by the same residualization map under $H_0$ and $H_1$, respectively.
\item[(T2)] \textbf{SRL/innovation regime restrictions:} the assumptions (A1)--(A4) in Section~\ref{sec:Model} hold; in particular, Willie performs warden-favorable cancellation using \(S\), and after conditioning on \(S\), the post-cancellation residual under \(H_0\) is  zero-mean Gaussian with known variance. Residual-variance uncertainty that remains non-negligible as \(n\) grows, as well as non-Gaussian post-cancellation residual laws, are therefore excluded from the achievability claim.
\item[(T3)] \textbf{Stationary innovation variance with a residual floor:} for all $s\in\mathcal S_{\mathrm{adm}}$, $\mathrm{Var}(R_{0,t}\,|\,S=s)=\sigma_0^2=\sigma_W^2+\sigma_e^2$ with $\sigma_e^2>0$.
\item[(T4)] \textbf{Vanishing covert power:} $\max_t P_t=o(\sigma_0^2)$.
\item[(T5)] \textbf{Arbitrary covert input:} conditioned on each admissible \(S=s\), the stationary specialization of (A3) holds, $R_0^n|S=s=N^n(s), \qquad R_1^n|S=s=X_c^n+N^n(s)$, with \(N^n(s)\sim\mathcal{N}(0,\sigma_0^2 I_n)\) independent of \(X_c^n\). The covert codeword distribution \(X_c^n\) may be arbitrary, including sparse or on-off signaling, subject to zero mean, per-use variances \(P_t\), and the vanishing-power condition.
\end{enumerate}

Then the maximal achievable covert message size satisfies
\begin{equation}\label{eq:main-throughput}
B_n^\star = R_C^\star \sqrt{n}\,(1+o(1)), \qquad n\to\infty,
\end{equation}
where the optimal first-order constant is
\begin{equation}\label{eq:main-constant}
R_C^\star =
\frac{\sigma_0^2}{\sigma_B^2\ln 2}\sqrt{\delta}.
\end{equation}
Moreover:
\begin{enumerate}[label={(B\arabic*)},ref={(B\arabic*)}]
\item[(B1)] \textbf{Achievability (attainment):} i.i.d.\ Gaussian covert symbols with the equalizer $P_n=2\sigma_0^2\sqrt{\delta/n}$ from~\eqref{eq:achi} achieve $B_n=R_C^\star\sqrt{n}(1+o(1))$ (Section~\ref{sec:Achievability}). \label{claim:ach}
\item[(B2)] \textbf{Converse:} under (T5), Gaussian residuals are least detectable for a fixed coordinate-variance profile. Therefore, any feasible scheme satisfying the additive Gaussian innovation model, including sparse or on-off signaling, obeys $\sum_{t=1}^n P_t^2\le4(1+o(1))\sigma_0^4\delta$ by Lemma~\ref{lem:Gaussian} and the argument in Section~\ref{sec:Converse}. Consequently, $B_n\le R_C^\star\sqrt n(1+o(1))$ by the Cauchy--Schwarz inequality and the AWGN information bound at Bob (Section~\ref{sec:Converse}). \label{claim:conv}
\item[(B3)] \textbf{Equality case / schedule structure:} among all power schedules meeting the relative-entropy budget, the uniform schedule is optimal under the relative-entropy budget for total power in the stationary case (Lemma~\ref{lem:uniform-opt} and Lemma~\ref{lem:equalizer}). \label{claim:eq}
\end{enumerate}

Equivalently, relative to the baseline constant $\bar R_C$ defined in Section~\ref{sec:residualization}, the stationary innovation-domain constant satisfies
\begin{equation}\label{eq:boost}
R_C^\star
=
\frac{\sigma_0^2}{\sigma_B^2\ln 2}\sqrt{\delta}
=
\left(1+\frac{\sigma_e^2}{\sigma_W^2}\right)\bar R_C,
\qquad
\sigma_0^2=\sigma_W^2+\sigma_e^2.
\end{equation}

\end{theorem}

\begin{proof}
(B1) follows directly from the achievability construction in Section~\ref{sec:Achievability}: i.i.d.\ Gaussian covert signaling with the equalizer $P_n = 2\sigma_0^2\sqrt{\delta/n}$ satisfies the conditioned relative-entropy budget and achieves $B_n = R_C^\star \sqrt{n}(1+o(1))$. 
(B2) follows from Section~\ref{sec:Converse}. Lemma \ref{lem:Gaussian} shows that, among all residual distributions with the same coordinate variances, the independent Gaussian variance-increase law minimizes relative entropy to the Gaussian null law. Thus sparse, on-off, or otherwise non-Gaussian signaling cannot relax the relative-entropy budget. The budget therefore implies the quadratic power bound \(\sum_t P_t^2\le 4(1+o(1))\sigma_0^4\delta\), and the AWGN information bound at Bob together with Cauchy--Schwarz gives \(B_n\le R_C^\star\sqrt n(1+o(1))\).
(B3) follows from Lemma~\ref{lem:uniform-opt} and Lemma~\ref{lem:equalizer} in Section~\ref{sec:kl-equalizer}, which show that in the stationary case the uniform schedule is optimal for total covert power under the relative-entropy budget and that the equalizer exists, is unique, and meets the relative-entropy budget with equality. 
Combining (B1) and (B2) yields the claimed first-order characterization $B_n^\star = R_C^\star \sqrt{n}(1+o(1))$.
\end{proof}

\section{Numerical Results and Discussion}\label{sec:results}

We numerically verify the relative-entropy equalizer, the first-order SRL constant, finite-\(n\) convergence, the \(\sigma^4\) scheduler, and the tradeoff between post- and pre-residualization designs under the assumptions of Sections~\ref{sec:Model}--\ref{sec:main}.

\subsection{Relative-Entropy Equalization Across \(n\)}

Figure~\ref{fig:equalizer-vs-n} shows that the exact equalizer \(\rho_n\), obtained from \(\frac{n}{2}(\rho_n-\ln(1+\rho_n))=\delta\), approaches the closed-form approximation \(2\sqrt{\delta/n}\) as \(n\) grows.

\begin{figure}[t]
\centering
\includegraphics[width=.9\linewidth]{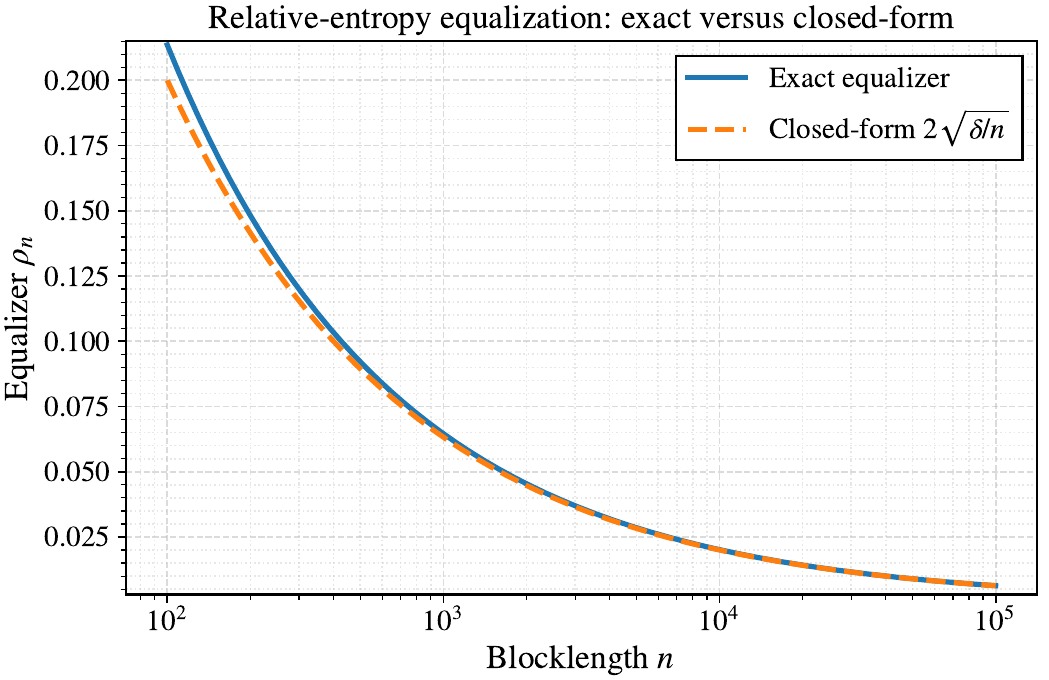}
\caption{Exact relative-entropy equalizer \(\rho_n\) versus \(2\sqrt{\delta/n}\) for \(\delta=1\). The two coincide for moderate \(n\), confirming the first-order accuracy of \eqref{eq:achi}.}
\label{fig:equalizer-vs-n}
\end{figure}

\subsection{First--Order Scaling in \(\delta\) and Constant Tightness}

Figure~\ref{fig:Bn-sqrt-n-vs-delta} confirms that both the exact and closed-form allocations approach the converse-tight first-order constant, with the small finite-\(n\) gap caused by higher-order terms.

\begin{figure}[t]
\centering
\includegraphics[width=.9\linewidth]{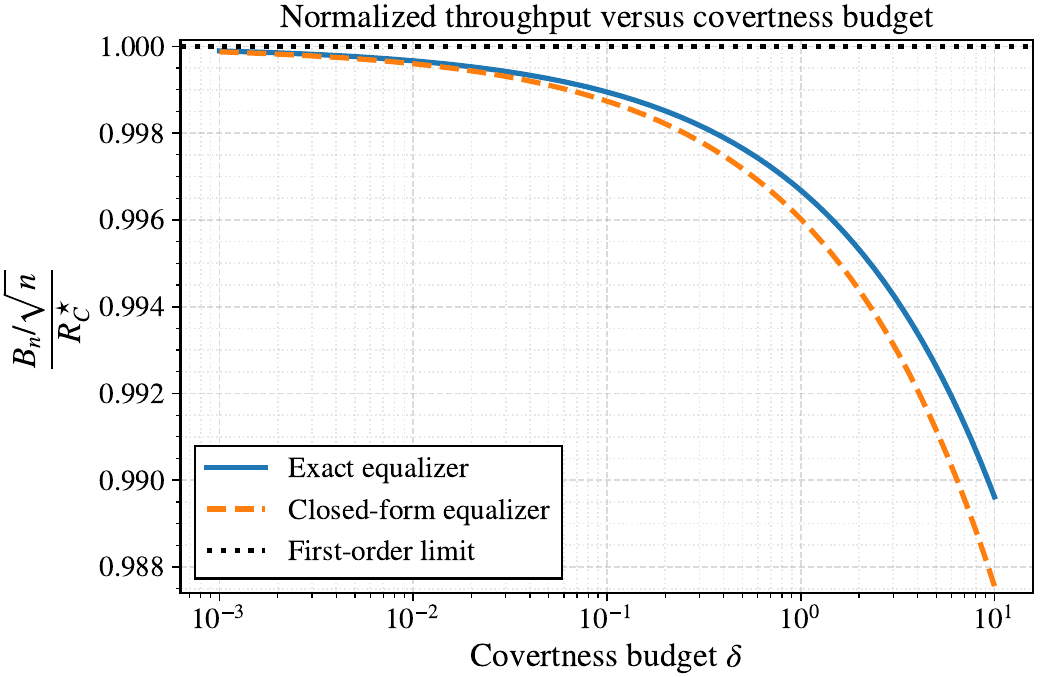}
\caption{Normalized throughput \((B_n/\sqrt{n})/R_C^\star\) versus \(\delta\) for \(n=10^6\), \(\sigma_B^2=1\), \(\sigma_W^2=2\), and \(\sigma_e^2=2\). Dotted: first-order limit; solid: exact equalizer; dashed: closed-form allocation.}
\label{fig:Bn-sqrt-n-vs-delta}
\end{figure}

\subsection{Finite--\(n\) Approach and Residual--Floor Boost}

Figure~\ref{fig:finite-n-approach} shows that increasing the residual floor raises the SRL constant while preserving the same asymptotic convergence behavior.

\begin{figure}[t]
\centering
\includegraphics[width=.9\linewidth]{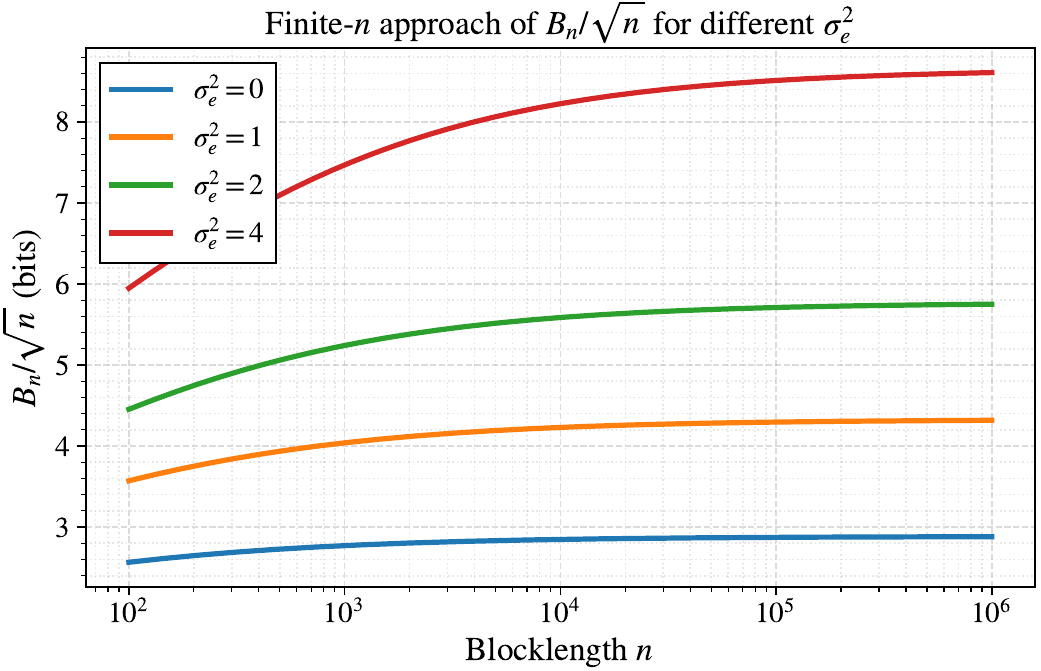}
\caption{Finite-\(n\) behavior of \(B_n/\sqrt{n}\) for several \(\sigma_e^2\), with \(\sigma_B^2=1\), \(\sigma_W^2=2\), and \(\delta=1\). Larger residual floors increase the limiting constant by \(1+\sigma_e^2/\sigma_W^2\).}
\label{fig:finite-n-approach}
\end{figure}

\subsection{Heterogeneous Innovation Variances: The \(\sigma^4\) Scheduler}

For time-varying conditioned innovation variances, the small-signal relative-entropy budget is $\sum_{t=1}^n \frac{P_t^2}{4\sigma_{0,t}^4}\lesssim \delta$. Maximizing the first-order total covert power under this constraint gives $P_t\propto\sigma_{0,t}^4$. For a known admissible variance profile consistent with (A3), Figure~\ref{fig:sigma4-schedule} shows that the gain over uniform allocation increases with heterogeneity and vanishes in the
stationary case.

\begin{figure}[t]
\centering
\includegraphics[width=.9\linewidth]{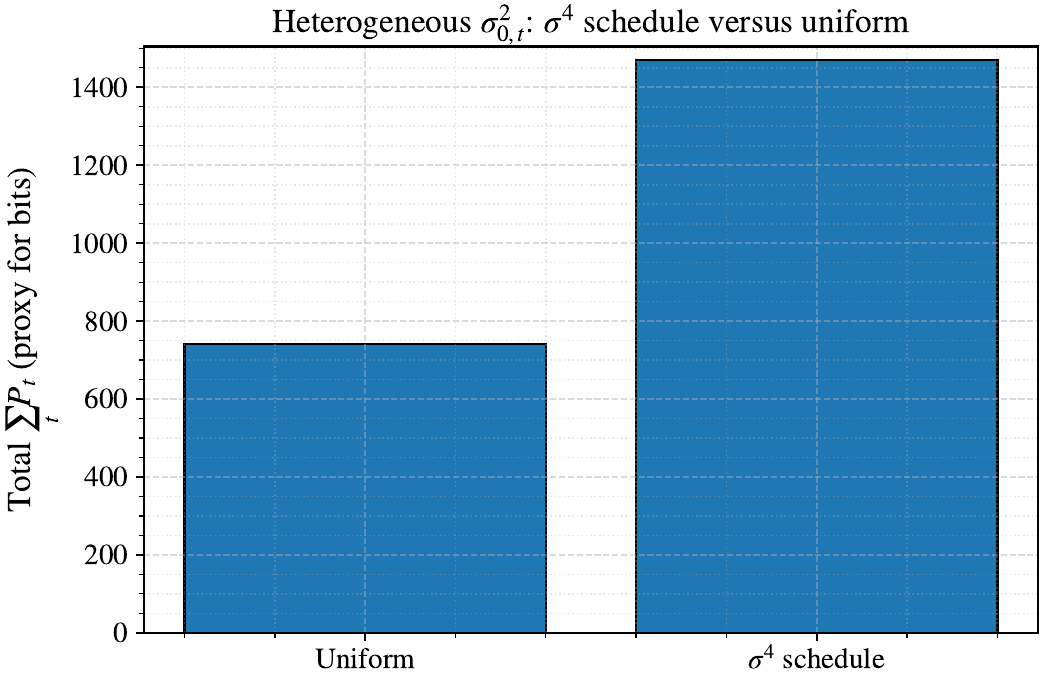}
\caption{\(\sigma^4\)-based versus uniform allocation for log-normal \(\{\sigma_{0,t}^2\}\), with \(n=3\times10^4\), \(\sigma_W^2=2\), \(\sigma_e^2=\sigma_W^2\), log-spread \(0.6\), and \(\delta=1\).}
\label{fig:sigma4-schedule}
\end{figure}

\subsection{Post- vs.\ Pre-Residualization Relative-Entropy Constraints}

Figure~\ref{fig:pre-vs-post} illustrates that the conservative pre-residualization design recovers only the baseline constant \(\bar R_C\), whereas post-residualization design exploits \(\sigma_0^2\) and attains \(R_C^\star\).

\begin{figure}[t]
\centering
\includegraphics[width=.9\linewidth]{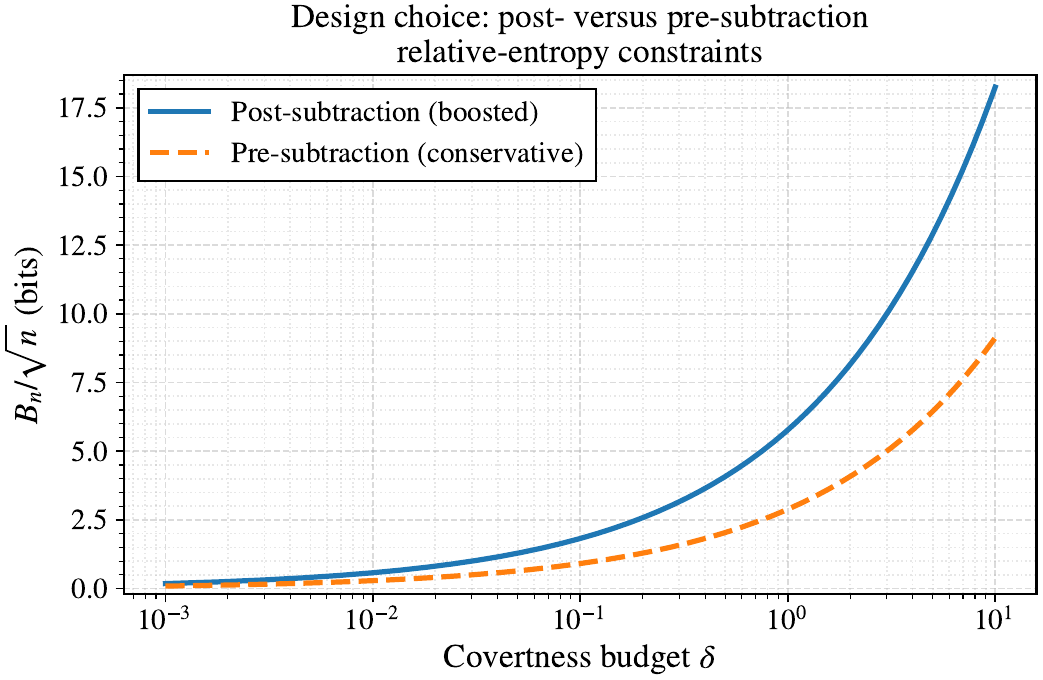}
\caption{Post- versus pre-residualization design for \(\sigma_B^2=1\), \(\sigma_W^2=2\), and \(\sigma_e^2=2\). The former attains \(R_C^\star\), while the conservative pre-residualization design attains \(\bar R_C\).}
\label{fig:pre-vs-post}
\end{figure}

\section{Future Directions}\label{sec:future}

The innovation-domain formulation and the SRL/linear-law boundary suggest several directions for future work.

\begin{enumerate}

\item \textbf{Adaptive innovation-domain design with estimated residual variance:} Several numerical results assume oracle knowledge of $\sigma_0^2$. A practical extension is an adaptive scheme that begins with the conservative pre-residualization constraint on the joint laws of $(W_i^n,S)$, $i\in\{0,1\}$, and approaches the innovation-domain design as pilot- and calibration-based confidence sets for $\sigma_R^2(s)$ tighten, while preserving the finite-$n$ guarantee $D\!\left(P_{R_1^n\mid S=s}\,\middle\|\,
P_{R_0^n\mid S=s}\right)\le\delta$ uniformly over admissible $s$.

\item \textbf{Origin and calibration of the innovation residual floor:} The parameter $\sigma_e^2$ represents the irreducible residue after warden-favorable cancellation. A practical objective is to estimate $\sigma_e^2$, or more generally $\sigma_R^2(s)$, with sufficient accuracy to preserve first-order SRL scaling. Pilot-based public-waveform cancellation, calibration measurements, and characterization of hardware, quantization, and tracking errors could provide direct evidence for the known Gaussian innovation condition in (A3).

\item \textbf{i.i.d.\ versus colored innovation residuals:} The present analysis relies on i.i.d.\ residuals. An extension would consider $R_0^n\sim\mathcal N(0,\Sigma_0),\qquad R_1^n\sim\mathcal N\!\left(0,\Sigma_0+\operatorname{diag}(P_t)\right)$, or low-rank covariance perturbations, together with eigenmode-aware power allocation and a converse under the same conditioned block relative-entropy constraint.

\item \textbf{Second-order characterization:} The main theorem is first-order tight, while the numerical results suggest an $O(n^{-1/2})$ convergence rate. A natural next step is a second-order analysis with explicit $O(1)$ terms and robustness penalties for online estimation of $\sigma_R^2(s)$.

\item \textbf{Operational characterization of the regime boundary:} An open problem is to express the SRL/linear-law boundary through measurable channel and protocol conditions. In particular, which properties of practical fading channels validate the known Gaussian innovation condition in (A3), and which residual-uncertainty mechanisms instead lead to linear-law or effective-secrecy regimes?

\end{enumerate}

\section{Conclusion}\label{sec:conclusion}
We studied covert communication in a warden-favorable innovation regime for an AWGN model with Bob and Willie observations. Willie uses all physically obtainable side information and performs block-level cancellation of the overt component before testing for
covert transmission. Under the conditioned post-subtraction relative-entropy constraint and the assumption that Willie’s post-cancellation residual is zero-mean Gaussian with known conditioned variance, the relevant null variance is
$\sigma_0^2=\sigma_W^2+\sigma_e^2$. In the stationary case, i.i.d.\ Gaussian covert signaling with $P_n=2\sigma_0^2\sqrt{\delta/n}$ achieves $B_n=R_C^\star\sqrt{n}(1+o(1))$, where $R_C^\star=\frac{\sigma_0^2}{\sigma_B^2\ln 2}\sqrt{\delta}$, and the converse establishes first-order optimality. In the heterogeneous case, the optimal first-order allocation satisfies $P_t\propto\sigma_{0,t}^4$, while imposing covertness before residualization yields a conservative design recovering the baseline constant $\bar R_C$. Thus, an irreducible post-cancellation residual floor changes the effective null variance and hence the SRL constant. If the Gaussian innovation model fails or residual-variance uncertainty remains non-negligible after conditioning, the present SRL results no longer apply, and noise-uncertainty or effective-secrecy mechanisms may lead to different, possibly linear-law, behavior.

\appendix
\phantomsection
\section*{Relative-Entropy Identities and Bounds}
\label{app:KL}

We collect the proofs of the relative-entropy expressions and inequalities used in Section~\ref{sec:kl-equalizer}. All statements apply conditioned on any admissible $S=s$ in the stationary innovation case, since the conditional laws are Gaussian with the stated variances.

\vspace{0.25em}
\subsection{One-Dimensional Variance Increase (Proof of~\eqref{eq:KLconvert}):} \label{app:KL_1}

For $\sigma_0^2 > 0$ and $P \ge 0$, with $\rho = P/\sigma_0^2$,
\begin{equation}
D\!\left(\mathcal N(0,\sigma_0^2+P)\,\middle\|\,\mathcal N(0,\sigma_0^2)\right)
= \frac12\bigl(\rho-\ln(1+\rho)\bigr).
\label{app:eq:KL-variance-increase}
\end{equation}
\emph{Proof.} Use $D(\mathcal N(0,\sigma_1^2)\|\mathcal N(0,\sigma_0^2))
=\tfrac12(\sigma_1^2/\sigma_0^2-1-\ln(\sigma_1^2/\sigma_0^2))$ with $\sigma_1^2=\sigma_0^2+P$. \hfill$\square$

\subsection{Upper and Lower Bounds for the Relative-Entropy Increment (Proof of~\eqref{eq:KLallt}):}\label{app:KL_2}

For $\rho \ge 0$,
\begin{equation}
\frac{\rho^2}{4(1+\rho)} \le \frac12\bigl(\rho-\ln(1+\rho)\bigr) \le \frac{\rho^2}{4}.
\label{app:eq:KL-increment-bounds}
\end{equation}
\emph{Proof.} $\ln(1+\rho) \ge \rho - \frac{\rho^2}{2}$ gives the upper bound; $\ln(1+\rho) \le \rho - \frac{\rho^2}{2(1+\rho)}$ gives the lower bound. \hfill$\square$

\subsection{Product Additivity (Proof of~\eqref{eq:KL-sum}):}\label{app:KL_3}

If $R_0^n$ and $R_1^n$ have product laws conditioned on $S=s$, with marginals $p_t$ and $q_t$, then
\begin{equation}
D\!\bigg(\prod_{t=1}^n p_t \,\bigg\|\, \prod_{t=1}^n q_t\bigg)
=\sum_{t=1}^n D(p_t\|q_t).
\label{app:eq:KL-product-additivity}
\end{equation}

\emph{Proof.} For product measures,
\begin{equation}
\log \frac{\prod_{t=1}^n p_t}{\prod_{t=1}^n q_t}
= \sum_{t=1}^n \log \frac{p_t}{q_t}.
\label{app:eq:log-product-ratio}
\end{equation}
Taking expectation with respect to $\prod_{t=1}^n p_t$ and interchanging sum and expectation yields~\eqref{app:eq:KL-product-additivity}. \hfill$\square$

\subsection{Blockwise Consequences (Clarifying Usage)}\label{app:KL_4}

For conditionally independent residual coordinates, combining \eqref{app:eq:KL-variance-increase} with \eqref{app:eq:KL-product-additivity} yields \eqref{eq:KL-sum}. Accordingly, \eqref{eq:KL-sum}--\eqref{eq:KL-lower-quad1} apply under the conditioned product-law residual model, with independence across \(t\) under both hypotheses given \(S=s\). The converse in Section~\ref{sec:Converse} does not require product additivity. It instead uses Lemma~\ref{lem:Gaussian}, which lower-bounds the relative entropy to the Gaussian null law for any zero-mean residual vector with the specified coordinate variances.


\end{document}